\documentclass[11pt]{article}
\usepackage[a4paper, total={6in, 9in}, top=1.5in]{geometry}
\usepackage[bitstream-charter]{mathdesign}
\usepackage[T1]{fontenc}
\usepackage{amsmath, amsthm}

\usepackage[all, knot]{xy}
\usepackage{verbatim}
\usepackage[english]{babel}
\usepackage{graphicx}
\usepackage[mdyyyy]{datetime}
\usepackage{algorithm}
\usepackage{algpseudocode}
\newcommand*\Let[2]{\State #1 $\gets$ #2}
\newcommand*\inlineif[3]{\State \textbf{if\ } #1 \textbf{\ then\ } #2 \textbf{\ else\ }#3\textbf{\ end if\ }}
\newcommand*\linif[2]{\State \textbf{if\ } #1 \textbf{\ then\ } #2 \textbf{\ end if\ }}
\algrenewcommand\alglinenumber[1]{{\sf\footnotesize#1}}
\algrenewcommand\algorithmicrequire{\textbf{Precondition:}}
\algrenewcommand\algorithmicensure{\textbf{Postcondition:}}

\newcommand{\Satisfiable}{\textsc{Satisfiable}}
\newcommand{\Checkconsistency}{\textsc{Checkbc}}
\newcommand{\Booleancase}{\textsc{Bool}}
\newcommand{\Modalcase}{\textsc{Modal}}
\newcommand{\emptymodel}{\ensuremath{\mathfrak{E}_0}}
\newcommand{\pmodel}[1]{\ensuremath{\mathfrak{#1}}}
\newcommand{\False}{\textbf{F}}
\newcommand{\True}{\textbf{T}}

\newcommand{\algReturn}{\State\textbf{Return\ }}
\newcommand{\M}{\ensuremath{\mathcal{M}}}

\newcommand{\rel}[1]{{\rightarrow_{#1}}}
\newcommand{\lr}[1]{\langle #1 \rangle}
\newcommand{\llrr}[1]{\{ #1 \}}

\newcommand{\mclass}[1]{\textbf{\textsf{#1}}}
\newcommand{\PSPACE}{\textbf{PSPACE}}
\newtheorem{theorem}{Theorem}
\newtheorem{definition}[theorem]{Definition}
\newtheorem{proposition}[theorem]{Proposition}
\newtheorem{lemma}[theorem]{Lemma}

\newcommand{\ELKv}{\textbf{ELKv}}
\newcommand{\LKv}{\textbf{LKv}}
\newcommand{\ELKvr}{\ensuremath{\ELKv^r}}

\newcommand{\LKvr}{\ensuremath{\LKv^r}}

\newcommand{\BP}{\textbf{P}}

\newcommand{\Ag}{\textbf{I}}

\newcommand{\SLKvr}{\ensuremath{\mathbb{LKV}^r}}

\newcommand{\mc}[1]{\mathcal{#1}}

\newcommand{\tr}[1]{\text{#1}}

\newcommand{\TAUT}{{\texttt{TAUT}}}

\newcommand{\RE}{{\texttt{RE}}}

\newcommand{\MP}{{\texttt{MP}}}

\newcommand{\DISTNSV}{\ensuremath{\texttt{DISTV}}}
\newcommand{\NSVBOT}{\ensuremath{\texttt{V}\bot}}
\newcommand{\NSVOR}{\ensuremath{\texttt{V}\lor}}

\newcommand{\NEC}{\ensuremath{\texttt{NEC}}}
\newcommand{\D}{\textbf{D}}

\newcommand{\N}{\ensuremath{\mathbb{N}}}

\newcommand{\lra}{\ensuremath{\leftrightarrow}}

\newcommand{\ncs}{\Box}
\newcommand{\psb}{\Diamond}
\newcommand{\nsv}{\nabla}
\newcommand{\imp}{\ensuremath{\rightarrow}}

\newcommand{\Lra}{\ensuremath{\Leftrightarrow}}
\newcommand{\ttt}[1]{\text{\texttt{#1}}}

\begin{document}

\title{Axiomatization and complexity of
  modal logic with knowing-what operator on model class \mclass{K}
  \footnote{The author is very grateful
    for Prof. Yanjing Wang's effort
    to provide much needed comments under his busy schedule.
    As this is still a draft,
    many things in the tableau part are yet unexplained,
    and the author begs the pardon of all readers.}}

\author{Yifeng Ding \\
  \small\textit{Group in Logic and the Methodology of Science, UC Berkeley}\\
  \small\texttt{yf.ding@berkeley.edu}\\ }
\date{}
\maketitle

\newpage

\begin{abstract}
    Standard epistemic logic studies propositional knowledge,
    yet many other types of knowledge such as
    ``knowing whether'', ``knowing what'', and ``knowing how''
    are frequently and widely used in everyday life as well as academic fields.
    An axiomatization of
    the epistemic logic with both
    regular ``knowing that'' operator and ``conditionally knowing what'' operator
    is recently given in
    [Yanjing Wang and Jie Fan.
    Conditionally knowing what.
    in \textit{Proceedings of AiML14}, April 2014.].
    Then the decidability and complexity of this logic command our study.
    In this paper,
    we give
    an axiomatization and a tableau
    for the modal logic with the same operators on arbitrary Kripke models.
    Given the tableau,
    the complexity of the satisfiability problem of this logic is \PSPACE-complete.
\end{abstract}
\newpage

\section{Introduction}

Standard epistemic logic studies the ``knowing that'' operator $K_i$ where
$K_i\phi$ means agent $i$ knows \textit{that} $\phi$ is true.
While this perspective fixed our focus on \textit{propositional knowledge},
its simplicity also facilitated the studies, extensions, and applications of it.
Recent decades witnessed the prosperity of numerous logics
with standard knowing-that operator or
similar propositional operators in fields like philosophy, computer science, and game theory.
However,
there are also other interesting knowledge expressions
used in our everyday life, like ``knowing whether'', ``knowing what'', and ``knowing how'',
which have raised many interesting questions in linguistics and philosophy,
but received less attention in logic.

Among these ways of expressing knowledge,
``knowing what'' is particularly suitable for
the beginning of our logical study of the myriad of non-standard knowledge operators,
for it is a richer topic compared to ``knowing whether'',
less contentious than ``knowing how'' philosophically,
interesting in its own logical and mathematical properties,
and readily applicable in other fields like cryptography.
For example, sentences like
``he knows that she knows her private key,
but he do not know what exactly his private key is.''
are typical in security settings.
With the propositional knowledge operator $K$ alone,
we may have a formula $K_iK_jp\land\lnot K_i p$ to express this.
But by axiom $\texttt{T}$ in standard epistemic logic,
this formula is not consistent.
Introducing something new is obviously needed,
and several attempts was made recently,
such as \cite{HP03:adversaries, DeciSP} in security settings.

In fact,
in his grounding work of the epistemic logic \cite{Hintikka:kab},
Hintikka has already briefly discussed ``knowing who'' in ch.6.3,
an operator with evident similarity with ``knowing what'',
in terms of first-order modal logic.
In \cite{Plaza89:lopc},
a seminal work that is hitherto mostly referred to by the studies of Public Announcement Logic,
Plaza also proposed a ``knowing what'' operator $Kv$,
of course in the context of Public Announcement Logic.
This leaves us a logic with both ``knowing what'' and public announcement.

Technically,
$Kv$ operator packs an existential quantifier with a modality together,
and the resulting logic is a small fragment of first-order modal logic,
which requires new techniques to handle.
To deal with the public announcement part,
we need to change our perspective and pack announcement into the ``knowing what'' operator
to make it a conditional one.
Thus until in \cite{WF14, WF13} by Wang and Fan did
we see a complete axiomatization of the logic with both the ``knowing what'' operator
and the model relativization operator, i.e., \ELKvr.
Because of the potential application of this logic,
such as in the field of computer science and AI as argued by McCarthy in \cite{McCarthy79},
the decidability and complexity of this logic become important.
In \cite{Xiong14}, Xiong has shown that \ELKvr is decidable for its small model property.
As for complexity, this paper serves as a preliminary step.

In this paper,
we show that the axiomatization of Wang and Fan without the characteristic S5 axioms
is also complete w.r.t. the logic on the class of arbitrary models
(call it \LKvr, that is, \ELKvr\ without the initial ``Epistemic'').
We simplifies the proof of completeness in \cite{WF14} significantly.
With the constraint of reflexivity,
there are some interactions between agents,
thus the beautiful property of the conditional part of knowing what operator in one agent is
obscured and complicated.
Without such constraint, we can work on the knowledge of an agent more easily and abstractly.

Moreover,
we show that the complexity of the satisfiability of the logic is \PSPACE-complete,
which is no more complex than most normal modal logics and in particular $K$.
This is by way of a tableau.
Normally a tableau means two things:
first, to test the satisfiability of a formula,
only its subformula counts,
and thus we can do trials on each of those subformulas by setting it true or false;
second, we have a canonical or minimal way to deal with the modal operators,
much like the spirit of Sahlqvist's minimal assignment method,
such that if this minimal way fails,
all possible ways fail necessarily.
As for our logic on the model class \mclass{K},
the first property is also true,
and for the second property,
there is not ``a'' canonical way but an array of them,
enumerable within \PSPACE.

The rest of this paper is structured as such:
we first give the syntax and semantics of \LKvr and its proof system \SLKvr in section 2.
Section 3 presents the completeness results and Section 4 the complexity.
We then conclude this paper with future work in Section 5.

\section{Preliminaries}
We follow the notations proposed in \cite{WF14}\ .
However, since we are now working on arbitrary Kripke models,
it is no longer appropriate to use $K$ as the modal operator.
So we now return to the box and diamond notation.

Given a countably infinite set of proposition letters $\BP$,
a countably infinite set of agent names $\Ag$,
and a countably infinite set of (non-rigid) constant symbols $\D$,
the language of \LKvr\ is defined as follows:
$$\phi ::= \top\mid p\mid \lnot\phi\mid (\phi\land\phi)\mid \ncs_i\phi\mid \nsv_i (\phi, d) $$
where $ p\in\BP, i\in\Ag$, and $d\in\D $.
Our new operator $\nabla_i(\phi,d)$ here says that,
in all possible cases where $\phi$ is true,
the value of $d$ is all the same.
For example,
the sentence
``I know your password if it is a four-digit number''
can be expressed as
$\nabla(\textbf{four-digit\_number\_password}, \textbf{password})$.
As usual, we define $\bot, (\phi\lor\psi), (\phi\imp\psi)$, and $ \psb_i\phi$
as the abbreviations of, respectively,
$\lnot\top, \lnot(\lnot\phi\land\lnot\psi), \lnot(\phi\land\lnot\psi)$, and
$\lnot\ncs_i\lnot\phi$.
Parentheses will be omitted unless confusion arises.

For future convenience,
write $Sub(\phi)$ for the set of subformulas of $\phi$,
where for $\nabla_i(\phi,d)$,
its subformulas are itself plus all the subformulas of $\phi$.
Then define $Sub^+(\phi) = \llrr{\lnot\phi\mid\phi\in Sub(\phi)}\cup Sub(\phi)$.
Let $D(\phi)$ be the set of the value names that occur in $\phi$.
At the same time, we need $depth(\phi)$ denoting the modal depth of $\phi$.
For the new operator $\nabla_i$,
we define $depth(\nabla_i(\phi,d)) = depth(\phi)+1$.
Further, for any finite set of formulas $X$:
    $$\begin{array}{llllll}
        Sub(X) & = & \bigcup_{\phi\in X} Sub(\phi) & \lnot X & = & \llrr{\lnot\phi\mid\phi\in X} \\
        Sub^+(X) & = & \bigcup_{\phi\in X} Sub^+(\phi) & X\backslash\Box_i & = & \llrr{\phi\mid\Box_i\phi\in X} \\
        D(X) & = & \bigcup_{\phi\in X} D(\phi) & \phi_X & = & \bigwedge_{\phi\in X}\phi \\
        depth(X) & = & \max\llrr{depth(\phi)\mid\phi\in X} &\Box_i X & = & \llrr{\Box_i\phi\mid\phi\in X}\\
    \end{array}$$

    To interpret \LKvr\ ,
    we need to extend common Kripke models to
    incorporate the assignment of the names in \D,
    and this can also be seen as a first-order Kripke model with a constant domain.
    So a model of \LKvr\ is defined as
    $\M  = \lr{S, O, \llrr{\rel{i}\mid i\in\Ag}, V, V_{\D}}$,
    in which $S$ is a non-empty set of possible worlds,
    $O$ is a non-empty set of values,
    $\rel{i}$ is a binary relation on $S$,
    $V$ is a function assigning to each proposition letter $p\in\BP$
    a set of possible worlds $V(p)\subset S$ where $p$ is true,
    and $V_{\D}$ a function from $\D \times S$ to $O$
    so that each value name $d\in \D$ at each possible world $s$ is assigned a value $V_{\D}(d,s)$.
    Let $\mclass{K}$ denote the class of all models defined above.
    Now the semantics:
    $$ \begin{array}{|lll|}
        \hline
        \mc{M},s \vDash \top  &      &   \textrm{always holds}\\
        \mc{M},s \vDash p     & \Lra &   s \in V(p) \\
        \mc{M},s \vDash \neg \phi &\Leftrightarrow& \mc{M},s\nvDash \phi \\
        \mc{M},s\vDash \phi\land \psi &\Leftrightarrow&\mc{M},s\vDash \phi \textrm{ and } \mc{M},s\vDash \psi \\
        \mc{M},s\vDash \ncs_i \phi &\Leftrightarrow&\text{for all } t \text{ such that }s\rel{i}t: \mc{M},t\vDash \phi  \\
        \M,s    \vDash \nsv_i(\phi,d)&\Leftrightarrow&\text{for any}~ t_1,t_2\in S \text{ such that }s\rel{i}t_1\text{ and }s\rel{i}t_2:\\&&\text{if } \M, t_1\vDash\phi \text{ and }\M,t_2\vDash\phi,\text{ then } V_\D(d, t_1)=V_\D(d,t_2)\\
        \hline
    \end{array}$$
    Intuitively,
    $\nabla_i(\phi,d)$ is true at $s$
    if and only if
    in all $i-$accessible worlds where $\phi$ is true,
    $d$ is assigned a uniform value.
    Conversely, for $\nabla_i(\phi,d)$ to be false,
    there must be two $i-$accessible $\phi-$worlds
    that disagree on the value of $d$.
    From the perspective of first-order modal logic,
    $\nabla_i(\phi,d)$ can be seen as
    $\exists x\Box_i(\phi\rightarrow d = x)$
    where $x$ is a rigid variable
    and $c$ a non-rigid one.
    Thus a $\nabla$ is actually a package
    consists of a quantifier, a modality, an implication, and an equality.

    As for the derivation system,
    it is enough to just exclude axioms particular to S5 from the system proposed in \cite{WF14}:\\
    \begin{minipage}{0.50\textwidth}
        \begin{center}
        \begin{tabular}{lc}
            \multicolumn{2}{l}{\textbf{System }\SLKvr}\\
            \multicolumn{2}{l}{Axiom Schemas}\\
            \texttt{TAUT} & \tr{ all the instances of tautologies}\\
            \texttt{K} &$ \ncs_i (\phi\to\psi)\to (\ncs_i \phi\to \ncs_i \psi$)\\
            \DISTNSV& $\ncs_i (\phi\to\psi)\to (\nsv_i(\psi,d)\to \nsv_i(\phi,d))$\\
            \NSVBOT& $\nsv_i(\bot,d)$\\
            \NSVOR& $\psb_i(\phi\wedge\psi)\wedge \nsv_i(\phi,d)\wedge \nsv_i(\psi,d)\to \nsv_i(\phi\vee\psi,d)$\\
        \end{tabular}
        \end{center}
    \end{minipage}
    \hfill
    \begin{minipage}{0.30\textwidth}
        \begin{center}
        \begin{tabular}{lc}
        \multicolumn{2}{l}{\ \ \ }\\
        \multicolumn{2}{l}{Rules}\\
        \MP& $\dfrac{\phi,\phi\to\psi}{\psi}$\\
        \NEC &$\dfrac{\phi}{\ncs_i \phi}$\\
        \RE &$\dfrac{\psi\lra\chi}{\phi\lra \phi[\psi\slash\chi]}$
        \end{tabular}
        \end{center}
    \end{minipage}

\section{Completeness}

Our proof of the completeness of \SLKvr proceeds in the standard Henkin way:
use maximal consistent sets as the basis of the canonical model,
link the canonical relations properly so that an existence lemma can be proven,
use the existence lemma to prove a truth lemma and then completeness follows immediately.
However,
as our $\nabla$ operator packs many things in it,
simply a maximal consistent set does not give us enough information
to pin down every possibilities.
Thus, we need to saturate these maximal consistent sets consistently.
Specifically, since $\nabla_i(\phi,d)$ is actually $\exists x\Box(\phi\rightarrow d = x)$,
its subformulas
$\Box(\phi\rightarrow d = x)$ and $d = x$ need their counterpart in the canonical model.
Now we give the definition:

    \begin{definition}\label{def.canonicalmodel}
        Denote the set of all maximal consistent sets w.r.t. \SLKvr as \textbf{MCS} and the set of natural number as \N. The canonical model $\M^c = \lr{S^c, \N, \llrr{\rel{i}\mid i\in \Ag}, V^c, V_{\D}^c}$ where\footnote{A countable set $\N$ is already big enough as the constant domain of objects. See footnote in \cite{WF14}. Also note that following clause (2) is slightly different from clause (ii) in \cite{WF14}.}:
        \begin{itemize}
            \item $S^c$ consists of all the triples $\lr{\Gamma, f, g} \in \text{\textbf{MCS}}\times \N^{\D}\times (\N\cup\llrr{*})^{\Ag\times \LKvr \times \D}$ that satisfy the following two conditions for any $ i\in \Ag, \phi, \psi\in\LKvr, d\in\D $:
            \begin{description}
                \item[(1)] $g(i, \phi, d) \not= *$ iff $\nsv_i(\phi, d)\land\psb_i\phi\in\Gamma$;
                \item[(2)] $g(i,\phi,d) \not= *$ and $g(i,\psi,d) \not= *$ imply: $g(i,\phi,d) = g(i, \psi, d)$ iff $\nsv_i(\phi\lor\psi, d)\in\Gamma$
            \end{description}
                for any $s\in S^c$, we use $\Gamma_s, f_s, g_s$ to denote the three components of $s$ and we simplify $\phi\in\Gamma_s$ as $\phi\in s$.
            \item For $s, t\in S^c$, $s\rel{i}t$ iff the following two conditions are satisfied:
            \begin{description}
                \item[(3)] $\llrr{\phi\mid\ncs_i\phi\in s} \subseteq t$.
                \item[(4)] $\nsv_i(\phi,d)\in s$ and $\phi\in t$ imply $f_t(d) = g_s(i,\phi,d)$.
            \end{description}
        \end{itemize}
    \end{definition}

    Here, $g$ is the counterpart of $\Box(\phi\rightarrow d = x)$ and $f$ the counterpart of $d = x$. To be explicit about their meaning, $g(i,\phi,d)$ gives the value of $d$ in all the $\phi-$worlds accessible by $i$, and $f(d)$ gives the value of $d$ directly. The star symbol obviously means that if there are no $\phi-$worlds accessible from $i$, $g(i,\phi,d)$ should reflect this fact by a value not in $\mathfrak{N}$.

    Given this canonical model, existence lemma is then our aim. In ordinary model logic, it is enough to use Lindenbaum lemma to extend $\Gamma\backslash\Box_i$ to build a $i-$successor of $\Gamma$. However, as our canonical model requires more information, or a saturation, we must show that such a saturation is possible i.e. is consistent with what we already have. The following proposition states this technically:

    \begin{proposition}\label{prop.construction}
        Given a possible world $s\in S^c$, an agent $i\in \Ag$, a maximal consistent set $\Gamma$ such that $\llrr{\phi\mid\ncs_i\phi\in s} \subseteq \Gamma$ and a natural number $x\in\N$, we can construct $t = \lr{\Gamma, f, g}$ using $x$ such that $t \in S^c$ and $s\rel{i}t$.
    \end{proposition}
    \begin{proof}
        Note that the only thing we need to do is to construct appropriate $f$ and $g$ so that $t = \lr{\Gamma, f, g}$ satisfies the requirements (1), (2) and (4) stated in definition \ref{def.canonicalmodel}, since (3) is already satisfied. We first construct $f$ (which is easier) and then $g$.

        For any $d\in \D$:
        $$ f(d) = \left\{
            \begin{array}{ll}
                g_s(i, \phi, d) & \text{there exists a\ } \phi\in\LKvr: \phi\in\Gamma\text{\ and\ }\nsv_i(\phi,d)\in s\\
                x & \text{otherwise}
            \end{array}
        \right.
        $$
        Obviously, if this $f$ is well-defined, then (4) in definition \ref{def.canonicalmodel} will be satisfied. Now we claim that this definition is indeed well-defined, that is, for any $\phi, \psi\in\LKvr$ such that $\phi\in\Gamma$, $\nsv_i(\phi,d)\in s$, $\psi\in\Gamma$ and $\nsv_i(\psi,d)\in s$, we have $g_s(i, \phi, d) = g_s(i, \psi, d)$.

        First, if $\phi\in\Gamma$ and $\psi\in\Gamma$, then $\psb_i(\phi\land\psi)\in s$. Suppose not, since $s$ is maximal, $\ncs_i(\lnot\phi\lor\lnot\psi)\in s$. Then $\lnot\phi\lor\lnot\psi \in\Gamma$. Again, since $\Gamma$ is maximal, either $\lnot\phi\in\Gamma$ or $\lnot\psi\in\Gamma$. But either way, $\Gamma$ will be inconsistent.

        Now $\psb_i(\phi\land\psi), \nsv_i(\phi,d), \nsv_i(\psi,d)$ are all in $s$. By axiom \NSVOR\ and the maximality of $s$, $\nsv_i(\phi\lor\psi, d)\in s$. According to clause (2) of definition \ref{def.canonicalmodel}, $g_s(i,\phi,d) = g_s(i,\psi,d)$ and this concludes the proof of the well-definedness of $f$.

        The construction of $g$ is more involved because of the clause (2). For any $i\in\Ag$ and any $d\in\D$, first we construct a partition on the set $G(i,d) = \llrr{\phi\in \LKvr\mid\nsv_i(\phi,d)\land\psb_i\phi\in\Gamma}$. Note that this set is exactly the collection of formulas that we need to give a non-star value through $g(i,\phi,d)$, and the clause (2) is effective only on this set. For any two $\phi,\psi\in G(i,d)$, let $\phi\sim_{i,d}\psi$ iff $\nsv_i(\phi\lor\psi,d)\in\Gamma$. Now we claim that $\sim_{i,d}$ is an equivalence relation:
        \begin{description}
            \item[Reflexivity]
                For $\phi\in G(i,d)$, by definition of $G(i,d)$, $\nsv_i(\phi,d)\in \Gamma$. By $\texttt{TAUT}$, $\phi\lra\phi\lor\phi$. By $\RE$ and the maximality of $\Gamma$, $\nsv_i(\phi\lor\phi,d)\in\Gamma$. So $\phi\sim_{i,d}\phi$.
            \item[Symmetry]
                For $\phi, \psi\in G(i,d)$, if $\phi\sim_{i,d}\psi$, then $\nsv_i(\phi\lor\psi,d)\in\Gamma$. By $\texttt{TAUT}$ and \RE, $\vdash\nsv_i(\phi\lor\psi,d)\lra\nsv_i(\psi\lor\phi,d)$. By maximality of $\Gamma$, $\nsv_i(\psi\lor\phi,d)\in\Gamma$. So $\psi\sim_{i,d}\phi$.
            \item[Transitivity]
                This is more interesting. For $\phi, \psi, \chi\in G(i,d)$, suppose $\phi\sim_{i,d}\psi$ and $\psi\sim_{i,d}\chi$. By definition, we have $ \psb_i\psi$, $\nsv_i(\phi\lor\psi,d)$ and $\nsv_i(\psi\lor\chi,d)$ all in $\Gamma$. Notice that we have following derivation: \\
                $\begin{array}{lll}
                \ttt{[1]}  &  \psi\to(\phi\lor\psi)\land(\psi\lor\chi)  &  \TAUT \\
                \ttt{[2]}  &  \psb_i\psi\to\psb_i((\phi\lor\psi)\land(\psi\lor\chi))  &  \text{K rule} \\
                \ttt{[3]}  &
                \!\!\!  \begin{array}{r}
                            \psb_i((\phi\lor\psi)\land(\psi\lor\chi))\land\nsv_i(\phi\lor\psi,d)\land\nsv_i(\psi\lor\chi,d) \\
                            \to\nsv_i((\phi\lor\psi)\lor(\psi\lor\chi),d)
                        \end{array}
                           & \NSVOR\\
                \ttt{[4]}  &  (\phi\lor\chi)\to(\phi\lor\psi)\lor(\psi\lor\chi) & \TAUT \\
                \ttt{[5]}  &  \Box_i((\phi\lor\chi)\to(\phi\lor\psi)\lor(\psi\lor\chi)) & \NEC \\
                \ttt{[6]}  &  \nsv_i((\phi\lor\psi)\lor(\psi\lor\chi),d)\to\nsv_i(\phi\lor\chi,d) & \DISTNSV, \ttt{[5]}\\
                \ttt{[6]}  &  \psb_i\psi\land\nsv_i(\phi\lor\psi,d)\land\nsv_i(\psi\lor\chi,d)\to\nsv_i(\phi\lor\chi,d) & \!\ttt{[2][3][6]}
                \end{array}$\\
                Since $\Gamma$ is a maximal consistent set w.r.t. \SLKvr, this tells us that $\nsv_i(\phi\lor\chi,d)\in\Gamma$. So $\phi\sim_{i,d}\chi$.
        \end{description}

        Write $[\phi]_{i,d} = \llrr{\psi\in G(i,d)\mid \psi \sim_{i,d} \phi}$. Since $\Ag, \BP, \D$ are all countable, \LKvr is countable, then $G(i,d)$ is countable, so $\llrr{[\phi]_{i,d}\mid \phi\in G(i,d)}$ is countable. Then there is an injection $h_{i,d}$ from $\llrr{[\phi]_{i,d}\mid \phi\in G(i,d)}$ to \N. Now we can define $g$:
        $$ g(i,\phi,d) = \left\{
        \begin{array}{ll}
            h_{i,d}([\phi]_{i,d}) & \phi\in G(i,d) \\
            * & \text{otherwise}
        \end{array}\right.
        $$
        It's now easy to see that this definition satisfies (1) and (2) of definition \ref{def.canonicalmodel}.
    \end{proof}

    For future convenience, we call this construction as $F$, that is, $F(s,\Gamma,a,x) = \lr{\Gamma, f, g}$ where $f$ and $g$ are defined as above.

    After the above proposition, we are now able to prove existence lemma. First is the existence lemma for $\lnot\Box_i\phi$:

    \begin{lemma}\label{lem.notncs}
        For any $s\in S^c$, any $i\in \Ag$, any $\phi\in\LKvr$: $\lnot\ncs_i\phi\in s$ implies that there is a world $t\in S^c$ such that $s\rel{i}t$ and $\lnot\psi\in t$.
    \end{lemma}
    \begin{proof}
        It is a standard modal logic exercise to show that $X = \llrr{\lnot\psi} \cup \llrr{\phi\mid\ncs_i\phi\in s}$ is consistent. By Lindenbaum Lemma (for \LKvr), $X$ can be extended into a \textbf{MCS} $\Gamma$. Then by proposition \ref{prop.construction}, $\Gamma$ can again be extended into a possible world $t = F(s,\Gamma,i,0) \in S^c$ such that $s\rel{i}t$. Since $\lnot\psi \in X$, $\lnot\psi\in t$.
    \end{proof}

    Now we need to deal with formulas in the form of $\lnot\nsv_i(\psi,d)$. Following the convention of dealing with $\lnot\ncs_i\psi$, what we need to do is to show that if $\lnot\nsv_i(\psi,d)$ is present in some possible world $s$, then there are indeed $t_1, t_2\in S^c$ such that $s\rel{i}t_1, s\rel{i}t_2$ and $f_{t_1}(d) \not= f_{t_2}(d)$. 
    More specifically:
    \begin{lemma}\label{lem.notnsv}
        For any $s\in S^c$ such that $\lnot\nsv_i(\psi,d)\in s$,
        there exists $t_1,t_2\in S^c$ such that $\psi\in t_1$, $\psi\in t_2$, $s\rel{i}t_1$, $s\rel{i}t_2$ and $f_{t_1}(d) \not= f_{t_2}(d)$.
    \end{lemma}
    \begin{proof}
        Suppose $s\in S^c$ and $\lnot\nsv_i(\psi,d)\in s$. Now we intend to prove\medskip\\
        (!) there exists $t_1,t_2\in S^c$ such that $\psi\in t_1$, $\psi\in t_2$, $s\rel{i}t_1$, $s\rel{i}t_2$ and $f_{t_1}(d) \not= f_{t_2}(d)$.

        Again we use the notation $G(i,d) = \llrr{\phi\in \LKvr\mid\nsv_i(\phi,d)\land\psb_i\phi\in s}, \sim_{i,d} = \llrr{\lr{\phi_1,\phi_2}\in G(i,d)^2\mid\nsv_i(\phi_1\lor\phi_2,d)\in s}$ and $[\phi]_{i,d} = \llrr{\phi'\in G(i,d)\mid\phi\sim_{i,d}\phi'}$ as defined in proposition \ref{prop.construction}. Let $A = \llrr{\phi\mid\ncs_i\phi\in s}$, $A^+ = A\cup\llrr{\psi}$,  $\overline{G(i,d)} = \llrr{\lnot\chi\mid\chi\in G(i,d)}$. Note that $A^+$ is consistent. Suppose it is not, then there is a finite subset $B$ of $A$ such that $\vdash\bigwedge B\to \lnot \psi$. By $\NEC$ and distribution of $\Box_i$, $\vdash\bigwedge\Box_i B\to\Box_i(\lnot\phi)$. Since $\Box_i B\subseteq s$, $\Box_i(\lnot\psi)\in s$. $\lnot\psi$ is equivalent to $\psi\to\bot$ and this means $\Box_i(\psi\to \bot)\in s$. By $\DISTNSV$ and $\NSVBOT$, $\nabla_i(\psi,d)\in s$, contradicting to supposition that $\lnot\nsv_i(\psi,d)\in s$.

        Now we prove (!) by two cases:

        \textbf{Case 1}: $A^+\cup\overline{G(i,d)}$ is consistent. Then $A^+\cup\overline{G(i,d)}$ can be extended by Lindenbaum Lemma to a maximal consistent set, say $\Gamma$. Let $\Gamma_1 = \Gamma_2 = \Gamma$, $t_1 = F(s,\Gamma_1,i,0)$ and $t_2 = F(s,\Gamma_2,i,1)$ and we have the following:
        \begin{itemize}
            \item $\psi\in t_1$, $\psi\in t_2$, $s\rel{i}t_1$, $s\rel{i}t_2$. By the construction method of $F$, this is immediate.
            \item $f_{t_1}(d) = 0$, $f_{t_2}(d) = 1$. From the construction rule of $f$ in proposition \ref{prop.construction}, we can see that these are true, by the fact that for all $\phi\in\LKvr$, either $\phi\not\in \Gamma$ or $\nabla_i(\phi,d)\not\in s$ and . In fact if $\nabla_i(\phi,d)\in s$, then $\phi\in G(i,d)$, $\lnot\phi\in \overline{G(i,d)}$. This means $\lnot\phi\in\Gamma$ and by the consistency of $\Gamma$, $\phi\not\in\Gamma$
        \end{itemize}
        With the above facts, the (!) is obviously true now.

        \textbf{Case 2}: $A^+\cup\overline{G(i,d)}$ is inconsistent. Then there is a finite subset $\overline{G(i,d)}_0$ of $\overline{G(i,d)}$ and a finite subset $A_0$ of $A$ such that $\vdash\bigwedge A_0\land\psi\to\lnot\bigwedge\overline{G(i,d)}_0$. Let $G(i,d)_0 = \llrr{\chi\mid\lnot\chi\in\overline{G(i,d)}_0}$. By the fact that $\vdash\lnot\bigwedge\overline{G(i,d)}_0\lra\bigvee G(i,d)_0$, we have $\vdash\bigwedge A_0\land\psi\to \bigvee G(i,d)_0$. For convenience, name this formula $\delta_0$

        At this point, we need to split case 2 into two subcases, with the following proposition as the dividing line:
        \medskip\\
        (*) for any $\chi_0\in G(i,d)$ there is a $\chi\in G(i,d)$ such that $\chi\not\in [\chi_0]_{i,d}$ and $A^+\cup\llrr{\chi}$ is consistent.

        \textbf{Case 2.1}: (*) is true. Since this still under Case 2, $A^+\cup\overline{G(i,d)}$ is inconsistent, which implies that there is a $\chi_1\in G(i,d)$ such that $A^+\cup\llrr{\chi_1}$ is consistent ($A^+$'s consistency is needed here). This implies, with (*), that there is a $\chi_2\in G(i,d)$ such that $\chi_2\not\in[\chi_1]_{i,d}$ and $A^+\cup\llrr{\chi_2}$ is consistent. The former means $\chi_1\not\sim_{i,d}\chi_2$, thus $\nsv_i(\chi_1\lor\chi_2,d)\not\in s$, which in turn means $g_s(i,\chi_1,d) \not= g_s(i,\chi_2,d)$ by the definition \ref{def.canonicalmodel}. Now since $A\cup\llrr{\chi_1}$ and $A\cup\llrr{\chi_2}$ are both consistent, let $\Gamma_1$ and $\Gamma_2$ be the \textbf{MCS}s extended by them respectively, and $t_1 = F(s,\Gamma_1,i,d)$ and $t_2 = F(s,\Gamma_2,i,d)$. It is not hard to see that $f_{t_1}(d)=g_{s}(i,\chi_1,d)\not=g_s(i,\chi_2,d)=f_{t_2}(d)$, which justifies (!).

        \textbf{Case 2.2}: (*) is false. Then the following\medskip\\
        (**) there exists a $\chi_0\in G(i,d)$ such that for any $\chi\in G(i,d)$, if $\chi\not\in [\chi_0]_{i,d}$ then $A^+\cup\llrr{\chi}$ is inconsistent.

        is true. Under this supposition, let $\chi_0$ be the element in $G(i,d)$ such that for any $\chi\in G(i,d)$. If $\chi\not\in[\chi_0]_{i,d}$ then $A^+\cup\llrr{\chi}$ is inconsistent. Further, let $[\chi_0]^0_{i,d} = G(i,d)_0\cap[\chi_0]_{i,d}$. Then, for any $\chi\in G(i,d)_0\setminus [\chi_0]^0_{i,d}$, $\chi \not\in [\chi_0]_{i,d}$, so $A^+\cup\llrr{\chi}$ is inconsistent, which means $\vdash \bigwedge A_0'\land\psi\to\lnot\chi$ (note it as $\delta_\chi$) for some finite subset $A_0'$ of set $A$. Combining $\vdash\delta_\chi$ for all $\chi\in G(i,d)_0\setminus [\chi_0]^0_{i,d}$, we have $\vdash\bigwedge A_0''\land\psi\to\lnot\bigvee(G(i,d)_0\setminus [\chi_0]^0_{i,d})$ again for some finite subset $A_0''$ of $A$. Note this long formula by $\delta_1$. Notice the following proof schema:
        $$\begin{array}{ll}
        (1) & \vdash \bigvee X\to(\bigvee Y\lor\bigvee (X\setminus Y))\\
        (2) & \vdash \bigvee X\to(\lnot\bigvee(X\setminus Y)\to\bigvee Y)\\
        (3) & \vdash (\bigvee X \land \lnot\bigvee(X\setminus Y))\to\bigvee Y
        \end{array}$$
        Using this schema, and the fact that $\vdash\delta_0$, $\vdash\delta_1$, we have the following proof:
        $$\begin{array}{ll}
        (4) & \vdash\bigwedge A_0\land\psi\to \bigvee G(i,d)_0 \qquad\qquad\qquad\quad [\text{this is }\delta_0]\\
        (5) & \vdash\bigwedge A_0''\land\psi\to\lnot\bigvee(G(i,d)_0\setminus [\chi_0]^0_{i,d}) \qquad [\text{this is }\delta_1]\\
        (6) & \vdash\bigwedge A_0\land \bigwedge A_0''\land\psi \to \bigvee G(i,d)_0 \land \lnot\bigvee(G(i,d)_0\setminus [\chi_0]^0_{i,d}) \\
        (7) & \vdash\bigwedge A_0\land \bigwedge A_0''\land\psi \to \bigvee[\chi_0]^0_{i,d} \\
        (8) & \vdash\bigwedge A_0\land \bigwedge A_0''\to (\psi \to \bigvee[\chi_0]^0_{i,d}) \\
        (9) & \vdash\ncs_i(\bigwedge A_0\land \bigwedge A_0'')\to \ncs_i(\psi \to \bigvee[\chi_0]^0_{i,d}) \\
        \end{array}$$
        By definition of $A$ and maximality, $\ncs_i(\bigwedge A_0\land \bigwedge A_0'') \in s$, so (***): $\ncs_i(\psi\to\bigvee[\chi_0]^0_{i,d})\in s$.

        Now we use a simple induction to show that $\nsv_i(\bigvee[\chi_0]^0_{i,d},d)\in s$. Enumerate the formula in $[\chi_0]^0_{i,d}$ as $\lambda_1, \lambda_2, \ldots \lambda_n$ and inductively define $\Lambda_1 = \llrr{\lambda_1}, \Lambda_k = \Lambda_{k-1}\cup\llrr{\lambda_k}$.
        \begin{description}
            \item[Induction Hypothesis] $g_s(i,\bigvee\Lambda_k,d) = g_s(i,\chi_0,d)$ and $\nsv_i(\bigvee\Lambda_k,d)\in s$.
            \item[Induction Basis] $\lambda_1\sim_{i,d}\chi_0$ so $\nsv_i(\lambda_1\lor\chi_0,d)\in s$, then $g_s(i,\lambda_1,d) = g_s(i,\chi_0,d)$. Since $\lambda_1\in G(i,d)$, $\nsv_i(\lambda_1,d)\in s$ automatically.
            \item[Induction Step] For $\Lambda_k = \Lambda_{k-1}\cup\llrr{\lambda_k}$, firstly, by the same kind of argument in induction basis, $g_s(i,\lambda_k,d) = g_s(i,\chi_0,d)$. By IH, $g_s(i,\chi_0,d) = g_s(i,\bigvee\Lambda_{k-1},d)$. So by the requirement (2) of a suitable possible canonical world in definition \ref{def.canonicalmodel} imposed on $s$, $\nsv_i(\bigvee\Lambda_{k-1}\lor\lambda_k,d) = \nsv_i(\bigvee\Lambda_k,d) \in s$. Since $\vdash\lambda_k\to\bigvee\Lambda_k$, $\vdash\psb_i\lambda_k\to\psb_i\Lambda_k$. Yet $\lambda_k\in G(i,d)$ so $\psb_i\lambda_k\in s$, then $\psb_i\bigvee\Lambda_k\in s$. Then both $g_s(i,\lambda_k,d)$ and $g_s(i,\Lambda_k,d)$ are not $*$. So by (2) of definition \ref{def.canonicalmodel} again, $g_s(i,\Lambda_k,d) = g_s(i,\lambda_k,d) = g_s(i,\chi_0,d)$.
        \end{description}
        By induction proof, $\nsv_i(\bigvee[\chi_0]^0_{i,d},d) = \nsv_i(\bigvee\Lambda_n,d)\in s$. Then with \DISTNSV\ and (***) we have proven, $\nsv_i(\psi,d)\in s$. But the proposition we intend to prove supposes $\lnot\nsv_i(\psi,d)\in s$. Thus this case 2.2 is actually empty.
    \end{proof}

    Now we are prepared to prove the truth lemma for $\M^c$:
    \begin{lemma}[truth lemma]
        For any $s\in S^c$ and any $\phi\in\LKvr$, $\phi\in s$ iff $\M^c,s\Vdash\phi$.
    \end{lemma}
    \begin{proof}
        The inductive proof of this is a common practice in modal logic. Here we only show the two non-trivial cases:
        \begin{description}
            \item[$\phi = \ncs_i\psi$] If $\ncs_i\psi\in s$, then for any $t\in S^c$ such that $s\rel{i} t$, by the clause (3) of definition \ref{def.canonicalmodel}, $\psi\in t$, which by IH means $\M^c,t\Vdash\psi$. So $\M^c,s\Vdash\ncs_i\psi$. For the other direction, suppose $\ncs_i\psi\not\in s$, then $\lnot\ncs_i\psi\in s$. By lemma \ref{lem.notncs} and IH, $\M^c,s\not\Vdash\ncs_i\psi$.
            \item[$\phi = \nsv_i(\psi,d)$] If $\nsv_i(\psi,d)\in s$, then for any $t_1, t_2\in S^c$ such that $s\rel{i}t_1,t_2$ and $\psi\in t_1, t_2$, by the clause (4) of definition \ref{def.canonicalmodel}, $f_{t_1}(d) = g_s(i,\psi,d) = f_{t_2}(d)$. For the other direction, suppose $\nsv_i(\psi,d)\not\in s$, then $\lnot\nsv_i(\psi,d) \in s$. By lemma \ref{lem.notnsv} and IH, we have $t_1,t_2\in S^c$ such that $s\rel{i}t_1,t_2$, $\M^c, t_1\Vdash\psi, \M^c,t_2\Vdash\psi$ and $f_{t_1}(d) \not= f_{t_2}(d)$. So $\M^c,s\not\Vdash\nsv_i(\psi,d)$.
        \end{description}
    \end{proof}

    Based on this, we are able to present:
    \begin{theorem}
        \SLKvr is sound and strongly complete for \LKvr.
    \end{theorem}
    \begin{proof}
        Soundness is rather simple. For any consistent set $\Delta\subseteq \LKvr$ , using Lindenbaum Lemma for \LKvr, there exists a \textbf{MCS} $\Gamma$ such that $\Delta\subseteq\Gamma$. Now let $f$ be a constant function from $\D$ to $0$, and $g$ be defined in the exactly same fashion as in proposition \ref{prop.construction}. According to definition \ref{def.canonicalmodel}, $s = \lr{\Gamma,f,g}\in S^c$, so by truth lemma, for any $\phi\in\Delta, \M^c,s\Vdash\phi$ and thus $\Delta$ is satisfiable. Then strong completeness follows.
    \end{proof}

\section{Complexity}

    In this section, we will give a \PSPACE\ algorithm in light of tableau method for the satisfiability problem of \LKvr. Since \LKvr contains \textbf{K}, the lower bound is also \PSPACE. So we can conclude that the decision problem of \LKvr is \PSPACE-complete.

    \subsection{Rules of tableau}

    \begin{definition}
        A propositional tableau is a set of formula $X$ satisfying the following:
        \begin{itemize}
            \item if $\lnot\lnot\phi\in T$ then $\phi\in T$,
            \item if $\lnot(\phi\land\psi)\in T$ then $\lnot\phi\in T$ or $\lnot\psi\in T$,
            \item if $\phi\land\psi\in T$ then $\phi\in T$ and $\psi\in T$,
            \item if $\phi\in T$ then $\lnot\phi\not\in T$ and vice versa,
        \end{itemize}
        We call a violation of the last clause ``blatantly inconsistent''. $X$ is fully expanded if and only if for any $\phi\in X$ and $\psi$ a subformula of $\phi$, either $\psi$ or $\lnot\psi$ is in $X$.
    \end{definition}

    \begin{definition}
        A state is a tuple $\lr{X, g, h, ha, hb}$ satisfying:
        \begin{itemize}
            \item $X$ is a fully expanded propositional tableau.
            \item Let $E_X = \llrr{\lr{i,d}\mid \text{for some}\ \phi, \nabla_i(\phi,d)\in X}$, $G_X(i,d) = \llrr{\phi\mid\nabla_i(\phi,d)\in X}$, $E_X(i) = \llrr{d\mid\lr{i,d}\in E_X}$.
            \item $g$ is a function defined on set $E_X$. $g(i,d)$ is a 2-tuple $\lr{A, B}$ such that:
                \begin{itemize}
                    \item $A\subseteq G_X(i,d)$, $B\subseteq \mathcal{P}(G_X(i,d))$;
                    \item $A\cup\bigcup B = G_X(i,d)$, $A\cap\bigcup B = \emptyset$;
                    \item $B$ is a partition of $\bigcup B$, always including empty set;
                \end{itemize}
                In the sequel let $g(i,d)[1]$ denote such $A$ and $g(i,d)[2]$ for such $B$.
            \item $h$ is a function defined on set $\llrr{\lr{i,\phi}\mid\lnot\Box_i\phi\in X}$. $h(i,\phi)$ is again a function defined on $E_X(i)$. For every $d\in E_X(i)$, $h(i,\phi)(d)\in g(i,d)[2]$.
            \item $ha$, $hb$ are both function defined on set $\llrr{\lr{i,\phi,d}\mid\lnot\nabla_i(\phi,d)}$. $ha(i,\phi,d)$ and $ha(i,\phi,d)$ are again functions defined on set $E_X(i)\cup\llrr{d}$ such that for $d'$ in their domain:
                \begin{itemize}
                    \item if $d'\in E_X(i)$, then $g(i,d')$ is defined, and $ha(i,\phi,d)(d')\in g(i,d')[2]$, $hb(i,\phi,d)(d')\in g(i,d')[2]$;
                    \item if $d'\not\in E_x(i)$, then $d' = d$. In this case $ha(i,\phi,d)(d') = hb(i,\phi,d)(d') = \emptyset$;
                    \item either $ha(i,\phi,d)(d) \not= hb(i,\phi,d)(d)$ or both of them are $\emptyset$.
                \end{itemize}
        \end{itemize}
    \end{definition}

    As we did in the proof of completeness, these functions $g, h, ha, hb$ are also ``extra information''. The function $g$ here is actually a enumeration of all possible equivalence relation $\sim_{i,d}$ given in the proof of proposition \ref{prop.construction}.

    It is worthwhile here to briefly discuss the number of possible $g, h, ha, hb$ for a given $X$. Obviously $|E_X|, |E_X(i)|, |G_X(i,d)| \le |X|$. For function $g$, note that $A$ and $B$ together forms a partition of $|G_X(i,d)|$. So the cardinality of the range of $g$ is at most $|X|^{|X|}$. Since the domain of $g$ is $E_X$, the cardinality of the domain of $g$ is at most $|X|$. Thus the total number of possible $g$ is at most ${|X|^{|X|}}^{|X|} = |X|^{|X|^2}$. Similarly, the number of all possible $h$, $ha$ and $hb$ are bounded by $|X|^{|X|}$. Summing all these together, given $X$, the number of all possible $\lr{g, h, ha, hb}$ is at most $|X|^{|X|^2 + 3 \times |X|}$.

    Now we present the method of deciding the satisfiability of a \LKvr formula $\phi_0$ trough building a tree. In the following rules, $L$ means the formula set of a node, $F$ represents the additional information needed ($g, h, ha, hb)$, and $C$ is a partial function from $D(\phi_0)$ to $\mathbb{Z}$ represents the required assignments of value names occurred in $\phi_0$. Since the set of all finite subsets of a countable set is also countable, there is a function, say, $code(X)$ to code each finite set of formulas into a unique positive integer.
    \begin{enumerate}
        \item Construct a tree with a single node $s_0$ as its root, and let $L(s_0) = \llrr{\phi_0}, F(s_0) = \emptyset, C(s_0) = \emptyset $.
        \item Repeatedly try each of following rules in their order until none of them applies:
        \begin{enumerate}
            \item \textbf{Forming propositional tableau}: if $s$ is a leaf node, $L(s)$ is not blatantly inconsistent and not a propositional tableau, then there must be a $\psi\in L(s)$ such that following 3 rules applies:
                \begin{enumerate}
                    \item if $\psi = \lnot\lnot\chi$, add a new node $s'$ and an edge between $s$ and $s'$ to the tree(i.e. a successor of $s$), and set $L(s') = L(s)\cup\llrr{\chi}, F(s') = F(s), C(s') = C(s)$.
                    \item if $\psi = \lnot(\chi_1\land\chi_2)$, add two successor $s_1, s_2$ of $s$, and set $L(s_i) = L(s)\cup\llrr{\lnot\chi_i}, F(s_i) = F(s), C(s_i) = C(s)$ for $i = 1, 2$.
                    \item if $\psi = \chi_1\land\chi_2$, add a successor $s'$ of $s$ and set $L(s') = L(s)\cup\llrr{\chi_1,\chi_2}, F(s') = F(s), C(s') = C(s)$.
                \end{enumerate}
            \item \textbf{Forming fully expanded propositional tableau}: if $s$ is a leaf node, $L(s)$ is a propositional tableau but not a fully expanded propositional tableau, then there must be $\phi\in Sub(L(s))$ such that $\phi$ and $\lnot\phi$ are both not in $L(S)$. In this case add two successor $s_1, s_2$ of $s$ and set $L(s_1) = L(s)\cup\llrr{\phi}$, $L(s_2) = L(s)\cup\llrr{\lnot\phi}$, $F(s_1) = F(s_2) = F(s), C(s_1) = C(s_2) = C(s)$.
            \item \textbf{forming state}: if $s$ is a leaf node, $L(s)$ is a fully expanded propositional tableau, but $\lr{L(s),F(s)}$ is not a state, then for all function tuple $F'$ such that $\lr{L(s), F'}$ is a state, add a successor $s'$ to $s$ and set $L(s') = L(s), F(s') = F', C(s') = C(s)$. Notice that the total number of such $F'$ is bounded by $|\phi_0|^{|\phi_0|^2 + 3\times |\phi_0|}$, as argued above.
            \item \textbf{Add labeled successors}: if $s$ is a leaf node, $\lr{L(s),F(s)}$ is a state and in $L(s)$ there are at least one formula of the form $\lnot\Box_i\phi$ or $\lnot\nabla_i(\phi,d_0)$, then there should be some labeled successors to $s$:
                \begin{itemize}
                    \item For each $\phi$ such that $\lnot\Box_i\phi\in L(s)$, add an $i$-successor(i.e. with an edge labeled $i$) $s'$ to $s$ and set $L(s') = $
                        $$\llrr{\lnot\phi}\cup L(s)\backslash\Box_i\cup\bigcup_{d\in E_{L(s)}(i)}\lnot\big(g_s(i,d)[1] \cup \bigcup g_s(i,d)[2]\backslash h_s(i,\phi)(d)\big)$$
                        $C(s') = h_s(i,\phi)$ and $F(s')$ all empty functions.
                    \item For each $\phi$ such that $\lnot\nabla_i(\phi,d_0)\in L(s)$, add two $i$-successor $s_a$ and $s_b$ to $s$ and for $\boldsymbol{x} = a, b$, set $L(s_{\boldsymbol{x}}) = $
                        $$\llrr{\phi}\cup L(s)\backslash\Box_i\cup\bigcup_{d\in E_{L(s)}(i)}\lnot\big(g_s(i,d)[1] \cup \bigcup g_s(i,d)[2]\backslash h\boldsymbol{x}_s(i,\phi,d_0)(d)\big)$$
                        Set $C(s_{\boldsymbol{x}}) = h\boldsymbol{x}_s(i,\phi,d_0)$ for $\boldsymbol{x} = a,b$. If $ha_s(i,\phi,d_0)(d_0) = hb_s(i,\phi,d_0)(d_0)$ then change $C(s_a)(d_0)$ to $\bullet$ and $C(s_b)(d_0)$ to $\circ$. Finally set $F(s_{\boldsymbol{x}}) = \emptyset$.
                \end{itemize}
            \item \textbf{Mark satisfiable}: if $s$ is not yet marked, non of the above three rules applies, and all its successors(possibly none) have been marked, then:
                \begin{itemize}
                    \item if the edges to the successors of $s$ are not labeled, then mark $s$ as "satisfiable" if any one of its successors is marked "satisfiable", otherwise mark "unsatisfiable".
                    \item if the edges to the successors of $s$ are labeled, then mark $s$ as "satisfiable" if all of its successors are marked "satisfiable", otherwise mark "unsatisfiable".
                    \item if $s$ has no successors, then mark $s$ as "satisfiable" if $L(s)$ is not blatantly inconsistent, otherwise mark "unsatisfiable".
                \end{itemize}
        \end{enumerate}
    \item if root $s_0$ is marked "satisfiable" then return $\phi_0$ is satisfiable, otherwise $\phi_0$ is unsatisfiable.
    \end{enumerate}

    \begin{lemma}
        For any \LKvr formula $\phi_0$, the tree construction method defined above terminates.
    \end{lemma}
    \begin{proof}
        It is immediate to see that if $s'$ is a successor of $s$ generated by rule (1) or (2) then $L(s)\subsetneq L(s')$ but for all $s$ in the tree, $L(s)\subseteq Sub^+(\phi_0)$. If $s'$ is generated from $s$ by rule (3), then rule (1) (2) and (3) are no longer applicable to $s'$. This means the longest chain of unlabeled edges will not exceed $2\times |\phi_0|+1$ otherwise there must be a blatant inconsistency. At the same time, if $s'$ is generated from $s$ by rule (4), then $depth(L(s')) < depth(L(s))$. Thus in any branch the number of labeled edges will not exceed $|\phi_0|$. So we can conclude that the depth of the tree is bounded by $2*|\phi_0|^2$. On the other hand, the branching number for any node is also bounded by $|\phi_0|^{|\phi_0|^2 + 3\times |\phi_0|}$. So this construction must terminate.
    \end{proof}

    After proving that this tableau must halt, the correctness of this tableau must be argued for now. Correctness means that, root $s_0$ is marked ``satisfiable'' if and only if $\phi_0$ is satisfiable. The following two lemmas present two directions of correctness respectively.

    \begin{lemma}
        For any \LKvr formula $\phi_0$, if after the tree construction defined above, root $s_0$ is marked ``satisfiable'', then $\phi_0$ is satisfiable.
    \end{lemma}
    \begin{proof}
        Suppose the root is marked ``satisfiable''. Then we can build a model satisfying $\phi_0$ from the constructed tree. Let $\M = \lr{W,O,\llrr{\rel{i}\mid i = 1 .. n},V,V_\D}$ where:
        \begin{itemize}
            \item $W = \llrr{s\mid s\ \text{is marked ``satisfiable'' and}\ \lr{L(s),F(s)}\ \text{is a state}}$;
            \item $O = $ all finite subset of \LKvr plus $\bullet$ and $\circ$;
            \item $s\rel{i}t$ if and only if there exists $s'\in W$ such that $s'$ is an $i$-successor of $s$ and $t$ is reachable from $s$ through a sequence of unlabeled edges;
            \item for all $s\in W$, if $p\in L(s)$ then $s\in V(p)$, if $\lnot p\in L(s)$ then $s\not\in V(p)$;
            \item for all $s\in W$, if $C(s)(d)$ is defined, then $V_\D(d,s) = C(s)(d)$.
        \end{itemize}
        By our construction method, there must be such a model. Now we can prove that if $\phi\in L(s)$ then $\M,s\Vdash\phi$ by a induction on $Sub^+(\phi_0)$. We give the key step of that induction:
        \begin{itemize}
            \item if $\nabla_i(\phi,d)\in L(s)$, then $d \in E_{L(s)}(i)$ and $\phi\in G_{L(s)}(i,d)$. Since $\lr{L(s), F(s)}$ is a state, $g_s(i,d)$ satisfies the clauses in the definition of state. Particularly, $\phi\in G_{L(s)}(i,d) = g_s(i,d)[1]\cup\bigcup g_s(i,d)[2]$.  Consider following two cases:
                \begin{itemize}
                    \item if $\phi\in g_s(i,d)[1]$, then by restraints on $g_s$ and $ha_s$ and rule (d), it is immediate that for all $i$-successors of $s$ $s'$, $\lnot\phi\in L(s')$. Thus for all $s''$ reachable from $s'$ through a sequence of unlabeled edges, $\lnot\phi\in L(s'')$. So if $s\rel{i}t$, $\lnot\phi\in L(t)$. By induction hypothesis, $\M,t\not\Vdash \phi$. Thus $\nabla_i(\phi,d)$ is trivially true on $s$.
                    \item if $\phi\in \bigcup g_s(i,d)[2]$, then there is a unique $X\in g_s(i,d)[2]$ such that $\phi\in X$. Now for any $t$ such that $s\rel{i}t$, by the property of $\rel{i}$, there exists $s'$ such that $s'$ is an $i$-successor of $s$ and $t$ is reachable from $s'$ through a sequence of unlabeled edges. By rule (d), $s'$ must be generated by a formula of the form $\lnot\Box_i\psi$ or $\lnot\nabla_i(\psi,d_0)$. W.l.o.g we suppose it is generated by $\lnot\nabla_i(\psi,d_0)$ and $ha_s$. If $\phi\in L(s')$, then $\lnot\phi\not\in L(s')$, because $s'$ must be marked ``satisfiable'' and thus is not blatantly inconsistent. Again by rule (d), $\phi\in ha_s(i,\psi,d_0)(d)$ for if not so, $\phi$ will in $\bigcup g_s(i,d)[2]\backslash ha_s(i,\psi,d_0)(d)$, then $\lnot\phi$ will be in $L(s')$, contradiction. By the constraints on $ha_s$, $ha_s(i,\psi,d_0)(d)$ must be $X$. Then by rule (d) again, $C(s')(d) = X$ and thus $C(t)(d) = X$. With this frame of argument, we can conclude that for all $t$ such that $s\rel{i}t$, if $\phi\in L(t)$ then $C(t)(d) = X$. By induction hypothesis ($\M,t\Vdash\phi$ implies $\phi\in L(t)$) and restraint on $V_\D$, we can conclude that $\M,s\Vdash\nabla_i(\phi,d)$.
                \end{itemize}
            \item if $\lnot\nabla_i(\phi,d)\in L(s)$, then it is immediate from rule (d) that there are two $i$-successor $s_a$ and $s_b$ such that $\phi\in L(s_a), \phi\in L(s_b), C(s_a)\not=C(s_b)$. Since $s$ is marked ``satisfiable'', $s_a$ and $s_b$ must also be so. By rule (e) and the finiteness of this tree, there must be $t_a$, $t_b$ in $W$ and reachable through a sequence of unlabeled edges from $s_a$ and $s_b$ respectively. Then $\phi\in L(t_a)$ and $\phi\in L(t_b)$ and $C(t_a)(d) = C(s_a)(d) \not= C(s_b)(d) = C(t_b)(d)$. By induction hypothesis and $\M$'s properties, $s\rel{i}t_a$, $s\rel{i}t_b$, $\M,t_a\Vdash\phi$, $\M,t_b\Vdash\phi$, $V_\D(d,t_a)\not=V_\D(d,t_b)$. So $\M,s\Vdash\lnot\nabla_i(\phi,d)$.
            \item If $\M,s\Vdash\nabla_i(\phi,d)$ then $\nabla_i(\phi,d)\in L(s)$. For suppose not, then $\lnot\nabla_i(\phi,d)\in L(s)$, then $\M,s\Vdash\lnot\nabla_i(\phi,d)$, contradiction. Similar results goes for $\lnot\nabla_i(\phi,d)$.
        \end{itemize}
        Since the root is marked ``satisfiable'', there must be a $s$ reachable through unlabeled edges from $s_0$ such that $s\in W$. Then $\phi_0\in L(s)$ and then $\M,s\Vdash\phi_0$, so $\phi_0$ is satisfiable.
    \end{proof}

    \begin{lemma}
        If $\phi_0$ is satisfiable, then after the construction for $\phi_0$, root $s_0$ will be marked ``satisfiable''.
    \end{lemma}
    \begin{proof}
        Through a induction from leaves to roots, we show that if $\lr{L(s),F(s)}$ is not a state and $L(s)$ is satisfiable, then $s$ is marked ``satisfiable''.

        First, if $s$ is a leaf, and $L(s)$ is satisfiable, then $L(s)$ must not be blatantly inconsistent. But since $s$ is a leaf, this suffices for $s$ to be marked ``satisfiable''.

        If $s$ is not a leaf and rule (a) or (b) was applied to $s$: w.l.o.g we show the case where (b) was applied to $s$, generating successor $s_1$ and $s_2$. Suppose both $L(s_1)$ and $L(s_2)$ are unsatisfiable, then by completeness theorem we have shown, they are inconsistent. So $\phi_{L(s)}\rightarrow\phi$ and $\phi_{L(s)}\rightarrow\lnot\phi$ are derivable. Thus $\phi_{L(s)}\rightarrow\bot$ is derivable, $L(s)$ is inconsistent. By soundness, $L(s)$ is unsatisfiable. Take a contraposition, we have if $L(s)$ is satisfiable, then either $L(s_1)$ or $L(s_2)$ is satisfiable. By induction hypothesis (note that $s_1$ and $s_2$ are not states), either $s_1$ or $s_2$ is marked ``satisfiable''. By rule (e), $s$ is marked ``satisfiable''.

        If $s$ is not a leaf and rule (c) was applied to $s$: suppose $L(s)$ is satisfiable, let $\M,s = \lr{W,O,\llrr{\rel{i}\mid i = 1,\ldots,n},V,V_\D},s$ be the model that satisfies $L(s)$. Now let $g$ be a function on $E_{L(s)}$ such that:
        \begin{itemize}
            \item $g(i,d)[1] = \llrr{\phi\in G_{L(s)}(i,d)\mid \text{for all}\ t\ \text{such that} s\rel{i}t, \M,t\not\Vdash\phi}$;
            \item $g(i,d)[2]$ is the partition of set $\llrr{\phi\in G_{L(s)}(i,d)\mid \text{there exists} t: s\rel{i}t, \M,t\Vdash\phi}$ defined by relation $\sim$ where $\psi_1\sim\psi_2$ if and only if there exists $t_1, t_2$ such that $s\rel{i}t_1$, $s\rel{i}t_2$, $\M,t_1\Vdash\psi$, $\M,t_2\Vdash\psi$, $V_\D(d,t_1) = V_\D(d,t_2)$. This $\sim$ relation is evidently a equivalence relation. Let $f(i,d,x)$ be the unique set $X\in g(i,d)[2]$ such that there exists $\psi\in X$ and $t\in W$ such that $\M,t\Vdash\psi$ and $V_\D(d,t) = x$. If there is no such a $X$ in $g(i,d)[2]$, let $f(i,d,x) = \emptyset$.
        \end{itemize}
        Then, let $h$ be a function on $\llrr{\lr{i,\psi}\mid\lnot\Box_i\psi\in L(s)}$. By supposition, $\M,s\Vdash\lnot\Box_i\psi$ for any $i,\psi$ in the domain of $h$. This means there exists $t\in W$ such that $s\rel{i}t$ and $\M,t\not\Vdash\psi$. Now let $h^*(i,\psi) = t$ and $h(i,\psi)$ be a function on $E_{L(s)}(i)$ such that $h(i,\psi)(d) = f(i,d,V_\D(d,t))$.
        Further, let $ha, hb$ be functions on $\llrr{\lr{i,\psi,d_0}\mid\lnot\nabla_i(\psi,d_0)\in L(s)}$. By supposition, $\M,s\Vdash\lnot\nabla_i(\psi,d_0)$ for any $\lr{i,\psi,d_0}$ in the domain of $ha$ and $hb$. This means there exists $t_a, t_b\in W$ such that both of them is accessible from $s$ through $i$, satisfies $\psi$ but $V_\D(d_0,t_a)\not=V_\D(d_0,t_b)$. Now for $\boldsymbol{x} = a, b$, let $h\boldsymbol{x}^*(i,\psi,d_0) = t_{\boldsymbol{x}}$ and let $h\boldsymbol{x}(i,\psi,d_0)$ be a function on $E_{L(s)}(i)\cup\llrr{d_0}$ such that $h\boldsymbol(i,\psi,d_0)(d) = f(i,d,V_\D(d,t_{\boldsymbol{x}}))$ for $d\in E_{L(s)}(i)$ and if $d_0\not\in E_{L(s)}(i)$, set $h\boldsymbol{x}(i,\psi,d_0)(d_0) = \emptyset$.

        Now it is evident that $\lr{L(s),g,h,ha,hb}$ is a state, so by rule (c), there will be a successor $s'$ of $s$ such that $F(s') = \lr{g,h,ha,hb}$. It is also not hard to see that after applying rule (d) to $s'$, for every successor $s''$ of $s'$, $L(s'')$ is satisfiable, because if $s''$ is generated by $h(i,\psi)$ (or $ha(i,\psi,d_0)$, $hb(i,\psi,d_0)$), then $\M,h^*(i,\psi)$ (or $ha^*(i,\psi,d_0), hb^*(i,\psi,d_0)$) $\Vdash L(s'')$.

        To see this more clearly, suppose $s''$ is generated by $ha(i,\psi,d_0)$ and let $t = ha^*(i,\psi,d_0)$. If $\alpha\in L(s'')$, then by rule (d), there are several cases:
        \begin{itemize}
            \item $\alpha = \psi$. By selection of $ha^*(i,\psi,d_0)$, this is evident;
            \item $\alpha \in L(s)\backslash\Box_i$. Since $s\rel{i}t$, this is also evident;
            \item for some $d \in E_{L(s)}(i)$, $\alpha\in \lnot g(i,d)[1]$. By definition of $g(i,d)[1]$, every $i$-accessible world from $s$ refutes every formula in $g(i,d)[1]$. So $\M,t\Vdash\alpha$.
            \item for some $d \in E_{L(s)}(i)$, $\alpha\in \lnot\big( \bigcup g(i,d)[2]\backslash ha(i,\psi,d_0)(d)\big)$. Let $\alpha = \lnot\beta$. Then $\beta$ does not belong to $ha(i,\psi,d_0)(d)$. Towards a contradiction suppose that $\M,t\Vdash\beta$, then by definition of $f$, $f(i,d,V_\D(d,t)) = [\beta]_{\sim}$. By definition of $ha$, $ha(i,\psi,d_0)(d) = [\beta]_{\sim}$, so $\beta\in ha(i,\psi, d_0)(d)$, contradiction. Thus, $\M,t\Vdash\lnot\beta$.
        \end{itemize}

        In conclusion, there is a successor $s'$ of $s$ such that for every successor $s''$ of $s'$, $L(s'')$ is satisfiable. By induction hypothesis, all such $s''$ is marked ``satisfiable''. By rule (e), $s'$ is marked ``satisfiable'', and so is $s$.
    \end{proof}

    It is straightforward to turn the above construction method into an algorithm running in polynomial space, using a depth-first search. For stepping down in the search tree, we need to record where we are currently by a stack where in every level a set of subformulas of $\phi_0$ is kept and the height of this stack is at most $|\phi_0|^2$. Thus we need $\mathcal{O}(|\phi_0|^2\times|\phi_0|)$ space. As the width of this tableau is exceedingly large, extra space is needed for branching. We need to enumerate all possible $F$ properly. At each level of the stack, we need to record where we are when enumerating $F$ so that the next $F$ can be calculated. This consumes $\mathcal{O}(|\phi_0|^2\times |\phi_0|^2)$ space. This means this algorithm runs in $\mathcal{O}(|\phi_0|^4)$ space, that is, in \PSPACE. Since this logic also contains modal logic \textbf{K}, its satisfiability problem is \PSPACE-hard. So we have theorem:
    \begin{theorem}
        The satisfiability problem for logic \LKvr is \textbf{PSPACE}-complete.
    \end{theorem}

\section{Conclusion}

In this paper,
we showed that $\SLKvr$ is sound and complete w.r.t. $\LKvr$ over arbitrary models
and gave a tableau for this logic.
This is just a start of the study of the complexity of similar ``knowing what'' logics.

Our proof of the completeness is relatively simpler than its counterpart in \cite{WF14}.
Exactly what makes this possible needs further investigation,
and we conjecture that, if this cause can be found,
we may give a beautiful frame of completeness proof upon which
proving completeness results on other special model classes will be easier.

Our tableau is not simple,
and more importantly,
unlike tableaux for normal modal logics
where if a formula is unsatisfiable,
a proof of its negation can be effectively constructed,
our tableau for \LKvr\ cannot provide this proof now.
This commands further study,
but our conjecture here is that,
a proof of the negation of an unsatisfiable formula is attainable from this tableau
or a slightly tweaked version,
even though it is not found yet.

The complexity of \ELKvr\ is what attracted us initially,
and our tableau may shed some light on it.
Yet it is still arguable whether it is in \PSPACE.
To make things more explicit,
we should try adding formulas
$d = x$ and $\Box_i(\phi\rightarrow d = x)$
directly into the tableau instead of using $G_i(\phi,d)$ and partitions,
which may only work on model class \mclass{K}.

Last but not least,
we should consider extending our language to incorporate more first-order characteristics,
such as predicate or equality.
If such extension does not bring too much complexity or other undesirable property,
we may also try to give a good logic on encryption,
as Cohen and Dam did in \cite{CoDa07}.

\bibliography{sgwyjj}

\end{document}